\documentclass[draftclsnofoot, onecolumn, letterpaper,
romanappendices]{IEEEtran}
\usepackage{Albert_Style}

\begin{document}
\title{Rateless Lossy Compression via the Extremes}

\author{
\IEEEauthorblockN{Albert No and Tsachy Weissman\\}
\IEEEauthorblockA{Department of Electrical Engineering, Stanford University\\
Email: \texttt{\{albertno, tsachy\}@stanford.edu}}
\thanks{The material in this paper has been presented in part at the 2014 52nd Annual Allerton Conference on Communication, Control, and Computing (Allerton).  This work was supported by the NSF Center for Science of Information under Grant Agreement CCF-0939370.}
}

\maketitle
\begin{abstract}
We begin by presenting a simple  lossy compressor operating at near-zero rate: The encoder merely describes the indices of the few maximal source components, while the decoder's reconstruction is a natural estimate of the source components based on this information. This scheme turns out to be near-optimal for the memoryless Gaussian source in the sense of achieving the zero-rate slope of its distortion-rate function. Motivated by this finding, we then propose a scheme comprised of iterating the above lossy compressor on an appropriately transformed version of the difference between the source and its reconstruction from the previous iteration. The proposed scheme achieves the rate distortion function of the Gaussian memoryless source (under squared error distortion) when employed on any finite-variance ergodic source. It further possesses desirable properties we respectively refer to as infinitesimal successive refinability, ratelessness, and complete separability. Its storage and computation requirements are of order no more than  $\frac{n^2}{\log^{\beta} n}$ per source symbol for $\beta>0$ at both the encoder and decoder. Though the details of its derivation, construction, and analysis differ considerably, we discuss similarities between the proposed scheme and the recently introduced Sparse Regression Codes (SPARC) of Venkataramanan et al.         

\end{abstract}
\IEEEpeerreviewmaketitle

\section{Introduction}\label{sec:Introduction}

Consider an independent and identically distributed (i.i.d.) standard Gaussian source $X^n = (X_1, X_2, \ldots, X_n)$. It is well known \cite{gnedenko1943distribution} that the maximum value concentrates on $\sqrt{2\log n}$, i.e., $\max_{1\leq i\leq n} X_i \approx \sqrt{2\log n}$. This fact suggests a simple lossy source coding scheme for the Gaussian source under quadratic distortion.
The encoder sends the index of the maximum value and the decoder reconstructs $\hat{X}^n$ according to 
\begin{align}
\hat{X}_i = \begin{cases}  \sqrt{2\log n}&\mbox{ if $X_i$ is the maximum}\\ -\frac{\sqrt{2\log n}}{n-1}&\mbox{ otherwise.}\end{cases}
\end{align}
For the meager $\log n$ nats that it requires, this simple scheme achieves essentially optimum distortion (in a sense made concrete in Section \ref{sec:Optimum Zero-Rate Gaussian Source Coding Scheme}) and has obviously modest storage and computational requirements. We can generalize this scheme by describing the indices of the $k_n$ largest values, and the scheme still achieves optimum distortion for its operating rate. Note that this scheme can be considered a special case of a permutation code \cite{berger1972permutation}, where the encoder sends a rough ordering of the source. It can perform as well as the best entropy-constrained scalar quantizer (ECSQ) but cannot achieve the optimum distortion-rate function at general positive rates \cite{goyal2001optimal}. In \cite{berger1972permutation}, the authors mentioned the $k_n=1$ case explicitly as being asymptotically optimum under the expected distortion criterion. Our focus is more on the excess distortion probability than the expected distortion. Furthermore, we establish a more general result where $k_n$ grows sub-linearly in $n$.

We generalize this idea to a scheme we refer to as Coding with Random Orthogonal Matrices (CROM), which achieves the distortion-rate function at all rates. Let $\Ab$ be a random $n$ by $n$ matrix uniformly drawn from the set of all $n$ by $n$ orthogonal matrices, i.e., for any $n$-dimensional vector $Y^n$, the random vector $\Ab  Y^n$ is uniformly distributed on the sphere with radius $\norm{Y^n}$. Since a random vector uniformly distributed on a
high-dimensional sphere is close in distribution to an i.i.d.\ Gaussian random vector, we can expect the behavior of $\Ab (X^n-\hat{X}^n)$ to be similar to that of an i.i.d.\ Gaussian random vector. Therefore, we can apply the above scheme again to describe a lossy version of it, using another $\log n$ nats, and so on. In this paper, we show that this iterative scheme achieves the Gaussian rate distortion function for any finite-variance ergodic source under quadratic distortion, while enjoying additional properties such as a strong notion of successive refinability and polynomial complexity.

One nice property of CROM is ratelessness. Similar to the rateless codes in the channel coding setting, CROM is able to reconstruct a source with partial messages while the optimum distortion for that rate is achieved. More precisely, suppose the decoder received first fraction $\nu$ of the messages for some $0<\nu<1$, then it can reconstruct a source with a distortion $D_G(\nu R)$. Thanks to the ratelessness, the encoder does not have to determine the rate ahead of encoding. However, unlike in many rateless channel coding settings, CROM requires that the bits observed are the first fraction $\nu$ bits, rather than that number of bits gleaned from any set of locations along the stream.

Much work has been dedicated to reducing the complexity of rate-distortion codes (cf. \cite{gioran2012complexity}, \cite{gupta2008rate}, \cite{korada2010polar}  and references therein). In particular, Venkataramanan et al. proposed the sparse regression code (SPARC) that achieves the Gaussian distortion-rate function with low complexity \cite{venkataramanan2014lossy}, \cite{venkataramanan2014lossy-minimum}. SPARC and CROM have similarities, which we discuss in detail.

The paper is organized as follows. In Section \ref{sec:Optimum Zero-Rate Gaussian Source Coding Scheme}, we present the simple zero-rate scheme and the sense in which it is optimal for  Gaussian sources. CROM is described, along with some of its properties and performance guarantees, in Section \ref{sec:Main Results}. We compare our scheme with SPARC in Section \ref{sec:Comparison to SPARC}. We test CROM via simulation in Section \ref{sec:Simulation Results}. We also discuss dual channel coding results in Section \ref{sec:Channel Coding Dual}. Section \ref{sec:Proofs} provides proofs of our main results and we conclude the paper in Section \ref{sec:Conclusions}.

\emph{Notation}: Both $X^n$ and ${\bf X}$ denote an $n$-dimensional random vector $(X_1,X_2,\ldots,X_n)$. We let $X_{(i)}$ denote the $i$-th largest element of $X^n$. We denote an $n$ by $n$ random orthogonal matrix by $\Ab$, and a non-random orthogonal matrix by $A$. We denote the distortion rate-function of the memoryless standard Gaussian source by $D_G(R)$. Finally, we use nats instead of bits and $\log$ denotes logarithm to the natural base unless specified otherwise.

\section{Optimum Zero-Rate Gaussian Source Coding Scheme}\label{sec:Optimum Zero-Rate Gaussian Source Coding Scheme}
 In this section, we propose a simple zero-rate lossy compressor which is essentially optimal for the i.i.d.\ standard Gaussian source under quadratic distortion. Before that, let us be more rigorous regarding our notion of ``zero-rate optimum source coding" for a Gaussian source under squared error distortion. Consider a scheme using a number of nats for the lossy description of the source which  is sub-linear in the block length $n$, i.e., the rate $R_n$ of the scheme converges to zero. Suppose the scheme achieves a distortion $D_n(\epsilon)$, where the target excess distortion probability is $\epsilon$, i.e., 
\begin{align}
\Pr{\frac{1}{n}\norm{X^n-\hat{X}^n}^2 >D_n(\epsilon)}<\epsilon.
\end{align}
We further define $D(n,0,\epsilon)$ to be the minimum distortion achievable over all possible strictly zero-rate schemes when the target excess distortion probability is $\epsilon$. Following lemma shows that the best reconstruction is the all zero vector ${\bf 0} = (0,0,\ldots,0)$ for the i.i.d.\ standard Gaussian source under squared error distortion. 
\begin{lemma}\label{lem:allzero}
Let $X^n$ be the i.i.d.\ standard Gaussian source. Then, for any $x^n\in \fR^n$ and $D>0$, the following inequality holds.
\begin{align}
\Pr{\norm{X^n-x^n}^2>D} \geq \Pr{\norm{X^n}^2>D}.
\end{align}
\end{lemma}
\begin{proof}
Since $X^n$ has spherically symmetric distribution, namely $AX^n$ is also i.i.d.\ standard Gaussian for any orthogonal matrix $A$, $\Pr{\norm{X^n-x^n}^2>D}$ only depends on $\norm{x^n}$. Let $\norm{x^n} = a$, then
\begin{align}
\Pr{\norm{X^n-x^n}^2>D} =& \Pr{(X_1-a)^2+\sum_{i=2}^n X_i^2>D}\\
=&\E{\Pr{(X_1-a)^2 > D- \sum_{i=2}^n X_i^2}\suchthat X_2,X_3,\ldots,X_n}\\
\geq&\E{\Pr{X_1^2 > D- \sum_{i=2}^n X_i^2}\suchthat  X_2,X_3,\ldots,X_n}\\
=& \Pr{\norm{X^n}^2>D}.
\end{align}
\end{proof}
Therefore, 
\begin{align}
D(n,0,\epsilon) \stackrel{\Delta}{=} \inf\left\{D:\Pr{\frac{1}{n}\norm{X^n}^2 >D}<\epsilon.\right\}.
\end{align}
It is not hard to show that
\begin{align}
D(n,0,\epsilon) = 1 + \sqrt{\frac{2}{n}} Q^{-1}(\epsilon) + O\left(\frac{1}{n}\right).
\end{align} 

Finally, we say that a sequence of zero-rate schemes achieves the zero-rate optimum if
\begin{align}
\lim_{n\rightarrow\infty} \frac{D_n(\epsilon)-D(n,0,\epsilon)}{R_n} = D_G'(0)
\end{align}
for all $\epsilon>0$, where $D_G'(0) = -2$ is the slope of the Gaussian distortion-rate function at zero rate. Equivalently,
\begin{align}
D_n(\epsilon) = 1  - 2 R_n+ \sqrt{\frac{2}{n}} Q^{-1}(\epsilon) + o\left(R_n\right).
\end{align}
This definition is reminiscent of the finite block length result in lossy compression \cite{DBLP:journals/corr/abs-1109-6310}, \cite{kostina2012fixed}, where the authors showed the minimum distortion $D(n,R,\epsilon)$ among all possible schemes for given rate $R$, target excess distortion probability $\epsilon$, and block length $n$ is
\begin{align}
D(n,R,\epsilon) = D_G(R) + \sqrt{\frac{2}{n}} Q^{-1}(\epsilon) + O\left(\frac{\log n}{n}\right).
\end{align}
Recall that $D_G(R)$ denotes the Gaussian distortion-rate function of memoryless standard Gaussian source.

We are now ready to propose the simple zero-rate optimum source coding scheme. Let $X^n = (X_1,X_2,\ldots, X_n)$ be an i.i.d.\ standard normal random process. The encoder simply sends the index of the maximum value, $m = \arg\max_{1\leq i \leq n} X_i$, and the decoder reconstructs $\hat{X}^n$ as 
\begin{align}
\hat{X}_i = \begin{cases}  \alpha_n&\mbox{ if $i=m$}\\ -\frac{\alpha_n}{n-1}&\mbox{ otherwise,}\end{cases}
\end{align}
where $\alpha_n>0$ is naturally chosen as $\E{X_{(1)}}\approx \sqrt{2\log n}$. Note that the encoder only describes the index of the maximum entry but not its value. This scheme works because the unsent value of the maximum entry concentrates on the specific value near $\sqrt{2\log n}$, i.e., $\max_{1\leq i\leq n} X_i \approx \sqrt{2\log n}$, which is
a well-known fact from extreme value theory \cite{gnedenko1943distribution}.

The rate of this scheme is $R_n = \frac{\log n}{n}$ nats per symbol, and it is not hard to show that the distortion is reduced by $2\frac{\log n}{n}$ (plus lower order terms), which is twice the rate we are using. Therefore, it is natural to suspect that such a scheme is zero-rate optimum.

We can generalize this scheme to send more than one index: The encoder sends the indices of the $k_n$ largest values of $X^n$, and the decoder reconstructs $\hat{X}^n$ as 
\begin{align}
\hat{X}_i = \begin{cases}  \alpha_n&\mbox{ if $X_i$ is one of the $k_n$ largest values of $X^n$}\\  -\frac{k_n\alpha_n}{n-k_n}&\mbox{ otherwise.}\end{cases}
\end{align}
Here we will choose $k_n = \lceil\log^{\beta} n\rceil$ for some $\beta>0$ and $\alpha_n$ to be roughly the expected value of the $k_n$-th largest value of $X^n$, i.e., $\alpha_n \approx \E{X_{(k_n)}}$. 

Clearly this scheme has rate $R_n = \frac{1}{n}\log {n\choose k_n}$ where $\lim_{n\rightarrow \infty} R_n = 0$. The following theorem shows that this scheme is optimal at zero rate.
\begin{theorem}\label{thm:optimum zero rate}
For any $\beta \geq 0$ and $k_n = \lceil\log^{\beta} n\rceil$, there is an $\alpha_n>0$ such that the above scheme achieves the zero-rate optimum. More precisely, for any $\epsilon>0$, the scheme achieves
\begin{align}
\Pr{\frac{1}{n}\norm{X^n-\hat{X}^n}^2>D_n}\leq \epsilon,
\end{align}
where 
\begin{align}
D_n = 1-2R_n + \sqrt{\frac{2}{n}} Q^{-1}(\epsilon) + O\left(\frac{k_n\log\log n}{n}\right).
\end{align}
Since $R_n = O\left(\frac{k_n\log n}{n}\right)$, we can say that the above scheme is zero-rate optimum.
\end{theorem}
The proof is given in Section \ref{sec:Proof of optimum zero rate in excess distortion probability}. We note that the encoding and decoding can be done in almost linear time. Moreover, we do not need to store an entire codebook, but only the single real number $\alpha_n$ needs to be stored.

\begin{remark}
Note that Verd\`{u} \cite{verdu1990channel} also considered the slope of the rate-distortion function at $D_{\max}$ as a counterpart to the capacity per unit cost. However, our requirements for zero-rate optimum scheme is more stronger since we incorporates the second order (or dispersion) term $\sqrt{\frac{2}{n}} Q^{-1}(\epsilon)$. 
\end{remark}

\begin{remark}
The above scheme only describes the index of the largest element. However, the encoder can send indices of both the maximum and the minimum, which is also the zero-rate optimum. Note that the minimum value will be close to $-\sqrt{2\log n}$, and therefore we can expect the similar behavior.
\end{remark}

\section{Coding with Random Orthogonal Matrices}\label{sec:Main Results}

\subsection{Preliminaries}
Before presenting the scheme, we briefly review some key ingredients: random orthogonal matrices and spherical distributions.

Let $\cO(n)$ be the set of all $n$ by $n$ orthogonal matrices. We write $\Ab \sim \mathUO$ to denote that $\Ab$ is a random $n$ by $n$ orthogonal matrix uniformly drawn from $\cO(n)$. This uniform distribution is with respect to Haar measure, cf. \cite{halmos1950measure}. More precisely, the random matrix $\Ab$ is uniformly distributed on $\cO(n)$ if and only if $B\times\Ab$ has the same distribution with $\Ab$ for any orthogonal matrix $B\in \cO(n)$. QR decomposition of random matrix with i.i.d.\ Gaussian entries provides a uniformly distributed random orthogonal matrix. There is a more efficient methods called subgroup algorithm to generate such matrices \cite{stewart1980efficient}, \cite{diaconis1987subgroup}. Now, let us recall the definition of a radially symmetric random vector and its relation with uniform random orthogonal matrices.
\begin{definition}
An $n$-dimensional random vector $X^n$ has a \emph{spherical distribution} if and only if $X^n$ and $AX^n$ has the same distribution for all orthogonal matrices $A\in \cO(n)$.
\end{definition}

One nice property of a spherically distributed random vector $X^n$ is that its characteristic function is radially symmetric \cite{schoenberg1938metric}, i.e., $\phi({\bf t}) = \E{\exp(i{\bf t}^TX^n)} = g(\norm{\bf t})$ for some $g(\cdot)$. Therefore, it is enough to consider the norm $\norm{X^n}_2^2$ for a spherically distributed random vector $X^n$. It is clear that an  i.i.d.\ Gaussian random vector has a spherical distribution. The following lemma shows how to \emph{symmetrize} a  vector with a uniform random orthogonal matrix.

\begin{lemma}\label{lem:RandomOrthogonalSphericalRandomVector}
Suppose $\Ab$ is a uniform random orthogonal matrix on $\cO(n)$. For any random vector $X^n$, the random vector $\Ab X^n$ has a spherical distribution.
\end{lemma}
The lemma is direct consequence of the respective definitions of a uniform random orthogonal matrix and a spherical distribution.

\subsection{Coding with Random Orthogonal Matrices}
For notational convenience, define $g_k:\fR^n \rightarrow \{0,1\}^n$ to be the function that finds the $k$ largest values of the input. If there is an ambiguity, the function picks the smallest index first.
Specifically, if $z^n =g_k(x^n)$, then $z_i=1$  if and only if $x_i$ is one of the $k$ largest entries of $x^n$ and $z_i=0$ otherwise. Let $A_1, A_2,\ldots,A_{L_n+1} \in \cO(n)$ be orthogonal matrices, $\alpha_1,\alpha_2,\ldots,\alpha_{L_n}$ be scalars, and assume that $k_n$ is a positive integer smaller than $n$. We are now ready to describe the iterative scheme.
\begin{algorithm}[H]
\begin{algorithmic}
\STATE Set $\Xb{1} = A_1 X^n$.
\FOR {$i=1$ to $L_n$}
\STATE Let $\mb{i} = g_{k_n}(\Xb{i})$.
\STATE Let $\Ub{i} = (\Ue{i}_1,\Ue{i}_2,\cdots,\Ue{i}_n)$ where
\begin{align}
\Ue{i}_j = 
\begin{cases}
\sqrt{\frac{n-k_n}{nk_n}} &\mbox{ if $ \me{i}_j=1$}\\
-\sqrt{\frac{k_n}{n(n-k_n)}}&\mbox{ otherwise}.
\end{cases}
\end{align}
\STATE Let $\Xb{i+1} = A_{i+1}(\Xb{i}-\alpha_i\Ub{i})$.
\ENDFOR
\STATE Send $(\mb{1},\mb{2},\ldots,\mb{L_n})$.
\end{algorithmic}
\caption{CROM}
\label{alg:RateDistortion}
\end{algorithm}

The unit vector $\Ub{i}$ indicates the $k_n$ largest values of $\Xb{i}$, and $\alpha_i$'s are scaling factors which depend on the norm of $\Xb{i}$ and will be specified later. Since $A_{i+1}^{-1} = A_{i+1}^{T}$, the inverse of the recursion is $\Xb{i} = A_{i+1}^T\Xb{i+1}+\alpha_i \Ub{i}$ for all $i$. This implies 
\begin{align}
X^n =  \alpha_1A_1^T\Ub{1} + \alpha_2 A_1^TA_2^T \Ub{2} + \cdots +\alpha_i (A_1^T\cdots A_i^T)\Ub{i} +  (A_1^T\cdots A_{i+1}^T)\Xb{i+1}.\label{eq:after ith iteration}
\end{align}
Therefore, when the decoder receives $(\mb{1},\mb{2},\ldots, \mb{i})$ for some $i\leq L_n$, it outputs the reconstruction 
\begin{align}
\Xbh{i} = \alpha_1A_1^T\Ub{1} + \alpha_2 A_1^TA_2^T \Ub{2} + \cdots +\alpha_i (A_1^T\cdots A_i^T)\Ub{i}.\label{eq:xhat after ith iteration}
\end{align}
The decoder can sequentially generate reconstructions using the relation $\Xbh{i+1} = \Xbh{i} + \alpha_i (A_1^T\cdots A_{i+1}^T)\Ub{i+1}$. Note that the decoder can compute $\Xbh{i}$ efficiently according to
\begin{align}
\Xbh{i} = A_1^T\left(\alpha_1\Ub{1} + A_2^T\left(\alpha_2 \Ub{2} + \cdots +\alpha_i  A_i^T\Ub{i}\right)\right).
\end{align}

Since we need $\log {n \choose k_n}$ nats to store (send) $\mb{i}$, rate $R$ corresponds to $L_n = \frac{nR}{\log{n\choose k_n}}$ number of iterations. We are ready to state our main theorem asserting that Algorithm \ref{alg:RateDistortion} achieves the Gaussian distortion-rate function.
\begin{theorem}\label{thm:infinite refinability}
Suppose $X^n$ is emitted by an ergodic source of marginal second moment $\sigma^2$.  For any $\beta\geq0$, let $k_n= \lceil(\log n)^{\beta}\rceil$ and suppose the rate is $R>0$. If we take
\begin{align}
\alpha_i = \sqrt{n\sigma^2\left(1-e^{-\frac{2}{L_n}R}\right)\left(e^{-\frac{i-1}{L_n}R}+e^{\frac{i-1}{L_n}R}\gamma_n\right)\left(e^{-\frac{i-1}{L_n}R}-e^{\frac{i-1}{L_n}R}\gamma_n\right)},
\end{align}
and small enough scalar $\gamma_n\equiv \gamma>0$, there exists orthogonal matrices $A_1,\cdots,A_{L_n+1}\in \cO(n)$ such that Algorithm \ref{alg:RateDistortion} satisfies
\begin{align}
\lim_{n\rightarrow \infty} \Pr{\frac{1}{n}\norm{X^n-\Xbh{i}}^2>\sigma^2\left(e^{-\frac{i}{L_n}R} + e^{\frac{i}{L_n}R}\gamma_n\right)^2\mbox{ for some $0\leq i\leq L_n$}}=0. \label{eq:infinite refinability thm}
\end{align}
Recall that \eqref{eq:infinite refinability thm} holds for any small enough $\gamma_n\equiv \gamma>0$ for any ergodic $X^n$. If we have stronger assumptions that $X^n$ is i.i.d.\ distributed with $\E{|X_1|^3}<\infty$, then we can find vanishing $\gamma_n = O\left(\frac{\log\log n}{\log n}\right)$ that satisfies \eqref{eq:infinite refinability thm}. 
\end{theorem}
The proof of Theorem \ref{thm:infinite refinability} is given in Section \ref{sec: proof of infinite refinability} with full details regarding the choice of $\gamma_n$. 
\begin{remark}\label{rem:tradeoff k}
Theorem \ref{thm:infinite refinability} implies that \eqref{eq:infinite refinability thm} holds for any fixed $\beta$. In terms of complexity, large $\beta$ is preferred since it implies small number of iteration which results in lower complexity. On the other hand, our result relies on the concentration of $k_n=\left(\log n\right)^{\beta}$ largest values of $n$ i.i.d.\ Gaussian random vector. If $\beta$ is too big, then the $k_n$ largest values may deviate too much. We will see the trade-off with simulation results in Section \ref{sec:Simulation Results}.
\end{remark}

\subsection{Discussion}\label{subsec:Discussion}

\subsubsection{Role of Orthogonal Matrices}\label{subsubsec:Role of Orthogonal Matrices}
It is known that an i.i.d.\ Gaussian random vector has a spherical distribution and the variance of its norm is very small. Therefore, if a random vector $X^n$ has a spherical distribution and the variance of its norm is small enough, $X^n$ can be thought of as an approximately i.i.d.\ Gaussian random vector. In the proof of CROM, we employ a randomization argument. Specifically, we assume that $\Ab_{1},\Ab_{2},\ldots,\Ab_{i+1}$ are drawn i.i.d.\ $\mathUO$ and show that equation \eqref{eq:infinite refinability thm} holds when the probability is averaged over this ensemble of random matrices. The source at $i$-th iteration $\Xb{i}=\Ab_i(\Xb{i-1}-\alpha_{i-1}\Ub{i-1})$ has spherical distribution by Lemma \ref{lem:RandomOrthogonalSphericalRandomVector}, and we can therefore expect $\Xb{i}$ to be a near Gaussian source, where we indirectly show that the norm of $\Xb{i}$ has small variance. This shows that multiplying by uniformly distributed random matrices can be thought of as a way to not only \emph{symmetrize} but also \emph{Gaussianize} the random vector so that we can apply the idea of Theorem~\ref{thm:optimum zero rate} iteratively.

Note that the conditional distribution of $\Ab X^n$ is no longer similar to Gaussian when the matrix $\Ab$ is known to both the encoder and the decoder. However, in the proof, we implicitly showed that the maximum element of ${\bf A}X^n$ is very close to $\sqrt{2\log n}$ with high probability as if it is i.i.d.\ Gaussian random vector.

A similar idea can be found in the work of Asnani et al.~\cite{asnani2013network}. The authors showed that any coding scheme for a Gaussian network source coding problem can be adapted to perform well for other network source coding problems that are not necessarily Gaussian but have the same covariances. The key idea of the paper is applying an orthogonal transformation to the sources which basically ``Gaussianizes" them so that the coding scheme for Gaussian sources are applicable in the transform domain.

\subsubsection{Storage and Computational Complexity}\label{subsubsec:Storage and Computational Complexity}
Unlike the zero-rate scheme of Section \ref{sec:Optimum Zero-Rate Gaussian Source Coding Scheme}, this scheme requires the storage of matrices (and scalars).
Since $L_n=  \left\lfloor\frac{nR}{\log {n\choose k_n}}\right\rfloor=O\left( \frac{n}{(\log n)^{\beta+1}}\right)$, both the encoder and decoder must keep $O\left(\frac{n^3}{\log^{\beta+1}n}\right)$
real values to store matrices $A_1, A_2, \ldots, A_{L_n}$. In terms of computation, the encoder finds the $k_n$ largest entries of an $n$ dimensional vector and performs a matrix-vector multiplication for each iteration.
The dominant cost is $O(n^2)$, the cost of matrix-vector multiplication. Therefore, the overall computational complexity is of order $O\left(\frac{n^3}{\log^{\beta+1}n}\right)$.

Instead of storing $A_1, A_2, \ldots, A_{L_n}$, it is also possible to store random seeds at both encoder and decoder to generate them. In this case, the CROM requires $O(1)$ storage space. However, generating a uniform random orthogonal matrix takes $O(n^3)$ \cite{stewart1980efficient}, and therefore the overall computational complexity will be of order $O\left(\frac{n^4}{\log^{\beta+1}n}\right)$.

\subsubsection{Infinitesimal Successive Refinability}\label{subsubsec:Infinitesimal Successive Refinability}
Suppose the decoder gets only the first $i$ messages $(\mb{1},\mb{2},\cdots,\mb{i})$. Note it needs to have seen only the first $n\frac{i}{L_n}R$ nats for that. With this partial message set, the decoder is able to reconstruct $\Xbh{i}$ which achieves a distortion
\begin{align}
\sigma^2\left(e^{-\frac{i}{L_n}R} + e^{\frac{i}{L_n}R}\gamma_n\right)^2, 
\end{align}
where the theorem guarantees $e^{\frac{i}{L_n}R}\gamma_n$ is arbitrarily negligible for large enough $n$. In other words, the decoder essentially achieves a distortion $\sigma^2 e^{-2\frac{i}{L_n}R}$, which is the Gaussian distortion-rate function at rate $\frac{i}{L_n}R$. Evidently, CROM can be viewed as a successive refinement coding scheme with $L_n$ stages. Since we have a growing number of stages (in $n$), the rate increment at each stage is negligible (i.e., sub-linear number of additional nats per stage) and this is a key difference from classical successive refinement problems where the number of stages is fixed. Note that Theorem \ref{thm:infinite refinability} implies that the probability of excess distortion beyond the relevant point on the distortion-rate curve at any of the successive refinement stages is negligible. Therefore, if the source is i.i.d.\ Gaussian, our coding scheme simultaneously achieves every point on the optimum distortion-rate curve. This \emph{infinitesimal successive refinability} can be considered a strengthened version of successive refinement. In other words, to implement and operate CROM, the value of the rate $R$ need not be known or set in advance, a point we will expound in Section \ref{subsubsec:(Near) Ratelessness}.

In \cite{ostergaard2011incremental}, the similar property called ``incremental refinements" was discussed. The paper discovered a new limiting behavior of additive rate-distortion function at zero-rate, and proposed a refinement idea. However, additive rate-distortion function is a mutual information between the input and the output of the Gaussian test channel, where it is not clear how to achieve it. On the other hand, we proposed a concrete scheme that achieves rate-distortion function.

\subsubsection{(Near) Ratelessness}\label{subsubsec:(Near) Ratelessness}
In the channel coding setting, it is well-known that rateless coding schemes, including Raptor codes, achieve the capacity of erasure channels. In this setting, the rate $R$ does not have to be specified in advance, and the receiver is able to decode a message upon observing sufficiently many packets (or bits), regardless of their order. As we have discussed above, CROM has a similar property in that a rate $R$ does not need to be specified in advance of the code design. This is because $\frac{R}{L_n}$ is a function of $n$ only, and therefore $\alpha_i$'s are independent to $R$. Furthemore, we will see in the proof that $\gamma_n$ depends only on $n$. If the source is i.i.d.\ $\cN(0,\sigma^2)$, the decoder can achieve a distortion $D_G(\nu R)$ upon observing fraction $\nu$ of the message bits. This is similar to a rateless code in channel coding because the decoder can achieve the optimum as soon as it collects sufficiently many of the message bits. However, the CROM decoder needs its observed bits to be a contiguous sequence at the beginning of the message bit stream while it is enough to have any combination of channel output observations in the rateless channel coding setting.

Note that our scheme can be considered as a progressive coder where ``progressive" refers to the refinability. However, it is often the case that the refinement layer of progressive code is often useless without the base layer, where refinement layers of CROM are useful by themselves. More precisely, the decoder can have the following reconstruction based only on $\mb{i_1},\ldots, \mb{i_l}$,
\begin{align}
{\bf \hat{X}} = \sum_{j=1}^{l} \alpha_{i_j} (A_1^T\cdots A_{i_j}^T)\Ub{i_j}
\end{align}
where with $\mb{1},\ldots,\mb{i_l}$ the reconstruction would be
\begin{align}
\Xbh{i_l} = \sum_{j=1}^{i_l} \alpha_{j} (A_1^T\cdots A_{j}^T)\Ub{j}.
\end{align}

\subsubsection{Complete Separability}\label{subsubsec:Complete Separability}
 In the classical separation scheme, the source encoder must know the channel capacity $C$ in order to design the source coding scheme with rate $R(D)<C$ where the source encoder often does not have this prior knowledge. However, if the source is Gaussian, the proposed scheme achieves the optimum distortion without channel information. Let $C_0$ be a sufficiently large constant and say the encoder uses the proposed scheme with rate $R= C_0$. When the decoder receives the first $C/C_0$ fraction of message bits and performs the reconstruction, we achieve the distortion $D$ that satisfies $R_G(D) = C$ due to the \emph{infinitesimal successive refinability}. Since we can achieve the optimum performance using a simple scheme while the source encoder is blind to the capacity of the link, we can call this property \emph{complete separability}.

\tikzstyle{format} = [thin]
\tikzstyle{medium} = [rectangle, draw, thin, fill=blue!20, minimum height=2.5em, minimum width = 4em]
\begin{figure}[H]
\centering
\begin{tikzpicture}[node distance=3cm, auto,>=latex', thick]
    \path[->] node[format] (src) {$X^n$};
    \path[->] node[medium, right of=src] (enc) {Encoder}
                  (src) edge node {} (enc);
    \path[->] node[medium, right of=enc] (relay) {Relay}
                  (enc) edge node {$C_1$} (relay);
    \path[->] node[medium, right of=relay] (dec) {Decoder}
                  (relay) edge node {$C_2$} (dec);
    \path[->] node[format, right of=dec] (rec) {$\hat{X}^n$}
                  (dec) edge node {} (rec);
\end{tikzpicture}
\caption{Relay Network} \label{fig:relay network}
\end{figure}
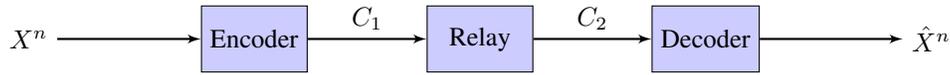

Another interesting example is a relay network without a direct link, as described in Figure \ref{fig:relay network}, where the source is i.i.d.\ Gaussian. Both the links from the encoder to the relay node and the relay node to the decoder are noiseless with capacity $C_1$ and $C_2$ respectively, when we assume that $C_1>C_2$. If the encoder knows the capacity of both links, then the problem is equivalent to the successive refinement problem. However, consider the case where the encoder only knows $C_1$. If the encoder is optimized only for the first link, the relay node has to decode the whole message and compress it again with rate $C_2$. However, if we use CROM, the relay node can simply send the first $\frac{C_2}{C_1}$ fraction of messages to the decoder and the decoder will be able to have optimal reconstruction with respect to its own link capacity.

\subsubsection{Convergence Rate}\label{subsubsec:Convergence Rate}
After the $i$-th iteration, the decoder can achieve a distortion 
\begin{align}
\sigma^2 \left(e^{-\frac{i}{L_n}R}+\gamma_n e^{\frac{i}{L_n}R}\right)^2 &= \sigma^2\left(e^{-2\frac{i}{L_n}R} + 2\gamma_n + e^{2\frac{i}{L_n}R}\gamma_n^2\right)\\
&\leq \sigma^2\left(e^{-2\frac{i}{L_n}R} + 2\gamma_n + e^{2R}\gamma_n^2\right).
\end{align}
Recall that the Gaussian distortion-rate function at rate $\frac{i}{L_n}R$ is $\sigma^2\exp\left(-2\frac{i}{L_n}R\right)$, and therefore the gap between the achieved distortion and $D_G\left(\frac{i}{L_n}R\right)$ is uniformly bounded by $2\sigma^2 \gamma_n+ \sigma^2e^{2R} \gamma_n^2$ at all stages. Note that if the source is i.i.d.\ with bounded $\E{|X_1|^3}$, we can choose vanishing $\gamma_n=O\left(\frac{\log\log n}{\log n}\right)$ such that the probability of error decays on the order of $O\left(\frac{1}{\log n}\right)$.

\section{Comparison to SPARC}\label{sec:Comparison to SPARC}
Recall that CROM can be viewed as a nonzero-rate generalization of the zero-rate scheme introduced in Section \ref{sec:Optimum Zero-Rate Gaussian Source Coding Scheme}.
On the other hand, SPARC implements the idea of describing a codeword with a linear combination of sub-codewords. Though the derivations of these two schemes were based on different ideas, they share several similarities. In this section, we outline the similarities and differences.

\subsection{Sparse Linear Regression Codes}\label{subsec:Sparse Linear Regression Codes}
Let us briefly review SPARC. Let $X^n$ be the first $n$ components of an ergodic source with mean 0 and variance 1. Define $L$  sub-codebooks $\cC_1,\cC_2,\ldots,\cC_L$, where each sub-codebook has $M$  sub-codewords. Sub-codewords are generated independently according to the standard normal distribution. Parameters $M$ and $L$ are chosen to be $M^L = e^{nR}$, where $R$ is the rate of the scheme, and define constants $c_1,c_2,\ldots, c_L$ appropriately. Then, the following algorithm exhibits the main structure of the sparse linear regression code (SPARC), which was presented in  \cite{venkataramanan2014lossy} and shown to achieve the Gaussian distortion-rate function for any ergodic source (under appropriate choice of parameters). 
\begin{algorithm}[H]
\begin{algorithmic}
\STATE Set $\Xb{1} = X^n$.
\FOR {$i=1$ to $L$}
\STATE Let $\Ub{i} = \argmax_{U^n\in \cC_i} <\Xb{i},U^n>$ and $\mb{i}$ be the index of $\Ub{i}$.
\STATE Let $\Xb{i+1} = \Xb{i}-c_i\Ub{i}$.
\ENDFOR
\STATE Send $(\mb{1},\mb{2},\ldots,\mb{L_n})$.
\end{algorithmic}
\caption{SPARC}
\label{alg:SPARCs}
\end{algorithm}
Note that there is another version of SPARC \cite{venkataramanan2014lossy-minimum} where encoding is not done sequentially but is done by exhaustive search. Since we are focusing on efficient lossy compressors, we only consider the SPARC described in Algorithm \ref{alg:SPARCs} throughout the paper. 

\subsection{Main Differences}\label{subsec:Main Differences}
In SPARC, the codebook consists of $L$  sub-codebooks where each sub-codebook has $M$ codewords. Our proposed iterative scheme is similar to SPARC with $L = \frac{nR}{\log n}$ and $M = n$; finding the sub-codeword that achieves the maximum inner product can be viewed as finding the maximum entries after multiplying the matrix in our iterative scheme. 

There are, however, two main differences. The first is that our scheme finds the $k_n$ largest values at each iteration. This implies that one iteration of our proposed encoding scheme is equivalent to $k_n$ iterations of SPARC's encoding. In Section \ref{subsubsec:Storage and Computational Complexity}, we have seen that CROM requires $O\left(\frac{n^2}{\log^{\beta+1} n}\right)$ operations per symbol, for an arbitrarily chosen $\beta>0$. The gap between the distortion and $D_G(R)$ is $\frac{\log \log n}{\log n}$. In SPARC, the gap between the distortion and $D_G(R)$ is $\frac{\log\log M}{\log M}$. In order to calibrate with CROM, we can set $M=n$. However, $ML$ operation per symbol is required for SPARC encoding where $M^L = e^{nR}$, and therefore the number of operations for SPARC is $O\left(\frac{n^2}{\log n}\right)$. Thus, SPRAC requires $\log^{\beta} n$ times more operations. The same relation holds when we consider the storage complexity. CROM requires to store $O\left(\frac{n^3}{\log^{\beta+1} n}\right)$ real numbers, where the SPARC encoder and decoder have to store $O\left(\frac{n^3}{\log n}\right)$ real numbers.

The second difference is the structure of the sub-codebook. The columns of orthogonal matrix are orthogonal to each other, and this implies that CROM is similar to SPARC with structured sub-codewords. For example, if $k_n=1$, all sub-codewords of CROM are orthogonal to each other, where SPARC draws sub-codewords according to i.i.d.\ Gaussian.

\subsection{Key Lemma}\label{subsec:Proofs}
As we discussed in Section \ref{subsec:Main Differences}, sub-codewords in CROM is drawn from the surface of the sphere while sub-codewords in SPARC are drawn according to the i.i.d.\ Gaussian distribution. Under this difference, we would like to introduce some dualities. For example, consider the following lemma used in the proof of SPARC.
\begin{lemma}\cite[Lemma 1]{venkataramanan2014lossy}
Let ${\bf Z}_1, \ldots, {\bf Z}_N$ be independent random vectors with i.i.d.\ standard Gaussian elements. Then for any random vector ${\bf B}$ supported on the $n$ dimensional unit sphere and independent of the ${\bf Z}_i$'s,  the inner products $\{ <{\bf Z}_i,{\bf B}>\}_{i=1}^N$ are i.i.d.\ standard Gaussian random variables that are independent of ${\bf B}$.
\end{lemma}

On the other hand, recall Lemma \ref{lem:RandomOrthogonalSphericalRandomVector}, which asserts that any random vector multiplied by uniform random orthogonal matrix has a spherical distribution. 

\subsection{Successive Refinability}\label{subsec:Successive Refinability}
That SPARC possesses the successive refinability property was briefly mentioned by the authors, however, the main theorem in \cite{venkataramanan2014lossy} only guarantees that the probability of error at the end of the process will vanish. On the other hand, we have seen that CROM has uniform convergence rates, uniformly and simultaneously on all points on the rate distortion curve, in Section \ref{subsubsec:Convergence Rate}.

\section{Simulation Results}\label{sec:Simulation Results}
In this section, we test CROM via simulations on sources with $\sigma^2=1$. We choose
\begin{align}
\alpha_i = \sqrt{n\left(1-e^{-\frac{2}{L_n}R}\right)}e^{-\frac{i-1}{L_n}R}.
\end{align}
Note that parameters are not optimized for the expected distortion, so there might be a better choice of $\alpha_i$. All results are averaged over 100 random trials.

First, We compare the performance of CROM and SPARC in Figure \ref{fig:Comparison to SPARC}. We choose i.i.d.\ standard Gaussian source $X^n$ where $n=256$. We simulated for $M = 128, 256, 512$ for SPARC. Note that the complexity of SPARC is higher when $M$ is large. We let $k_n=1$ for CROM which corresponds to $M=256$ case of SPARC. Note that the performance of CROM is similar to the performance of SPARC with $M=256$.

As we discussed in Remark \ref{rem:tradeoff k}, the complexity of CROM decreases when $k_n$ is large, however, the performance will be worse when $k_n$ is large. Figure \ref{fig:different k n1024} shows trade-off between the small and the large $k_n$.

In order to simulate CROM with higher $n$, we use structured orthogonal matrices to reduce the storage and computational complexity. Note that any orthogonal matrix is a product of $\frac{n(n-1)}{2}$ Givens rotations which are matrices of the form

\begin{align}
\begin{pmatrix}
1&\cdots&0&\cdots&0&\cdots&0\\
\vdots& \ddots&\vdots& &\vdots& &\vdots\\
0&\cdots&\cos \theta&\cdots&-\sin \theta&\cdots&0\\
\vdots& &\vdots& &\vdots& &\vdots\\
0&\cdots&\sin\theta&\cdots&\cos \theta&\cdots&0\\
\vdots& &\vdots& &\vdots& \ddots&\vdots\\
0&\cdots&0&\cdots&0&\cdots&1\\
\end{pmatrix}.
\end{align}

This suggests to construct sparse orthogonal matrices using Givens rotations as a building block. Suppose $n$ be the power of $2$, i.e., $n = 2^s$. We recursively define the sparse orthogonal matrices $A_r^{(s)}$ for $1\leq r\leq s$.
\begin{align}
A_r^{(s)}(\theta_1,\ldots,\theta_{n/2}) &= \begin{cases}
\begin{pmatrix} A_{r-1}^{(s-1)}(\theta_1,\ldots,\theta_{n/4})
&{\bf 0}\\
{\bf 0}
&A_{r-1}^{(s-1)}(\theta_{n/2+1},\ldots,\theta_{n/2})\end{pmatrix}&\mbox{ if $r>1$}\\
\begin{pmatrix}
\mbox{diag}(\cos\theta_1,\ldots,\cos\theta_{n/2})
&\mbox{diag}(-\sin\theta_1,\ldots,-\sin\theta_{n/2})\\
\mbox{diag}(\sin\theta_1,\ldots,\sin\theta_{n/2})
&\mbox{diag}(\cos\theta_1,\ldots,\cos\theta_{n/2})
\end{pmatrix} &\mbox{ if $r=1$}
\end{cases},
\end{align}
where $\mbox{diag}(x_1,\ldots,x_n)$ is a diagonal matrix with entries $x_1,\ldots,x_n$. The following matrices \eqref{eq:matrix_type1}, \eqref{eq:matrix_type2}, \eqref{eq:matrix_type3} show three types of sparse orthogonal matrices when $n=8$.
\begin{scriptsize}
\begin{align}
A_1^{(3)}(\theta_1,\theta_2,\theta_3,\theta_4) = 
\begin{pmatrix}
\cos \theta_1 &0&0&0& -\sin\theta_1 & 0&0&0\\
0&\cos \theta_2 &0&0&0& -\sin\theta_2 &0&0\\
0&0&\cos \theta_3 &0&0&0& -\sin\theta_3 &0\\
0&0&0&\cos \theta_4 &0&0&0& -\sin\theta_4 \\
\sin \theta_1 &0&0&0& \cos\theta_1 & 0&0&0\\
0&\sin \theta_2 &0&0&0& \cos\theta_2 & 0&0\\
0&0&\sin \theta_3 &0&0&0& \cos\theta_3 & 0\\
0&0&0&\sin \theta_4 &0&0&0& \cos\theta_4  \\
\end{pmatrix}\label{eq:matrix_type1}\\
A_2^{(3)}(\theta_1,\theta_2,\theta_3,\theta_4) = 
\begin{pmatrix}
\cos \theta_1 &0& -\sin\theta_1 & 0&0&0&0&0\\
0&\cos \theta_2 &0& -\sin\theta_2 &0&0&0&0\\
\sin \theta_1 & 0&\cos\theta_1 & 0&0&0&0&0\\
0&\sin \theta_2&0 & \cos\theta_2 & 0&0&0&0\\
0&0&0&0&\cos \theta_3 & 0& -\sin\theta_3 &0\\
0&0&0&0&0&\cos \theta_4 &0& -\sin\theta_4 \\
0&0&0&0&\sin \theta_3 &0& \cos\theta_3 & 0\\
0&0&0&0&0&\sin \theta_4 &0& \cos\theta_4  \\
\end{pmatrix}\label{eq:matrix_type2}\\
A_3^{(3)}(\theta_1,\theta_2,\theta_3,\theta_4) = 
\begin{pmatrix}
\cos \theta_1 & -\sin\theta_1 & 0&0&0&0&0&0\\
\sin \theta_1 & \cos\theta_1 & 0&0&0&0&0&0\\
0&0&\cos \theta_2 & -\sin\theta_2 &0&0&0&0\\
0&0&\sin \theta_2 & \cos\theta_2 & 0&0&0&0\\
0&0&0&0&\cos \theta_3 & -\sin\theta_3 &0&0\\
0&0&0&0&\sin \theta_3 & \cos\theta_3 & 0&0\\
0&0&0&0&0&0&\cos \theta_4 & -\sin\theta_4 \\
0&0&0&0&0&0&\sin \theta_4 & \cos\theta_4  \\
\end{pmatrix}\label{eq:matrix_type3}.
\end{align}
\end{scriptsize}
Each matrix $A^{(r)}_s$ is a product of $n/2$ Givens rotations. Therefore, the product of $\log n$ consecutive sparse orthogonal matrices is equivalent to the product of $\frac{n\log n}{2}$ Givens rotations. If we draw angles uniformly randomly, the product is expected to have similar distribution to uniform random orthogonal matrix. Since each row has exactly two non-zero elements, the matrix multiplication requires $O(n)$ operations. Also, the storage complexity is $O(n)$.

Another well-known orthogonal matrix is discrete cosine transform matrix of type-II (DCT-II). We can use Fast Fourier Transform (FFT) algorithm to multiply DCT matrix efficiently. Also, DCT matrix requires $O(1)$ of storage space. 

Instead of original CROM with uniform random orthogonal matrices, we propose two modified version of CROM using the above structured orthogonal matrices. First, at $i$-th iteration, we choose $A_{r}^{(s)}(\theta^{(s)}_{r,1},\ldots,\theta^{(s)}_{r,n/2})$ where $i \equiv r$ (mod $s$), and $\theta^{(s)}_{r,1},\ldots,\theta^{(s)}_{r,n/2}$ are uniformly sampled from $[0,2\pi]$. The second approach is using $A^{\mbox{DCT}} \times A_{r}^{(s)}(\theta^{(s)}_{r,1},\ldots,\theta^{(s)}_{r,n/2})$ where $A^{\mbox{DCT}}$ denotes the DCT-II matrix. Figure \ref{fig:structured_matrices} shows performances of two modified algorithms when $n=1024$ and $k_n=1$. Note that the performance of sparse orthogonal matrices is worse than uniformly generated orthogonal matrices, on the other hand, the performance of sparse orthogonal matrices with DCT-II matrix is comparable to those of uniform orthogonal matrices.

Since modified CROM has lower complexity, we can test CROM with larger $n$. Figure \ref{fig:different n} shows the distortion-rate curve of the second approach with sparse orthogonal matrices and the DCT-II matrix where $n=65536$ and $k=1$. Compare to the simulation result of $n=1024$ with uniform random orthogonal matrices, its distortion-rate curve shows better performance.

\begin{figure}
\centering
\begin{subfigure}[t]{0.45\textwidth}
\centering
\includegraphics[width = \textwidth]{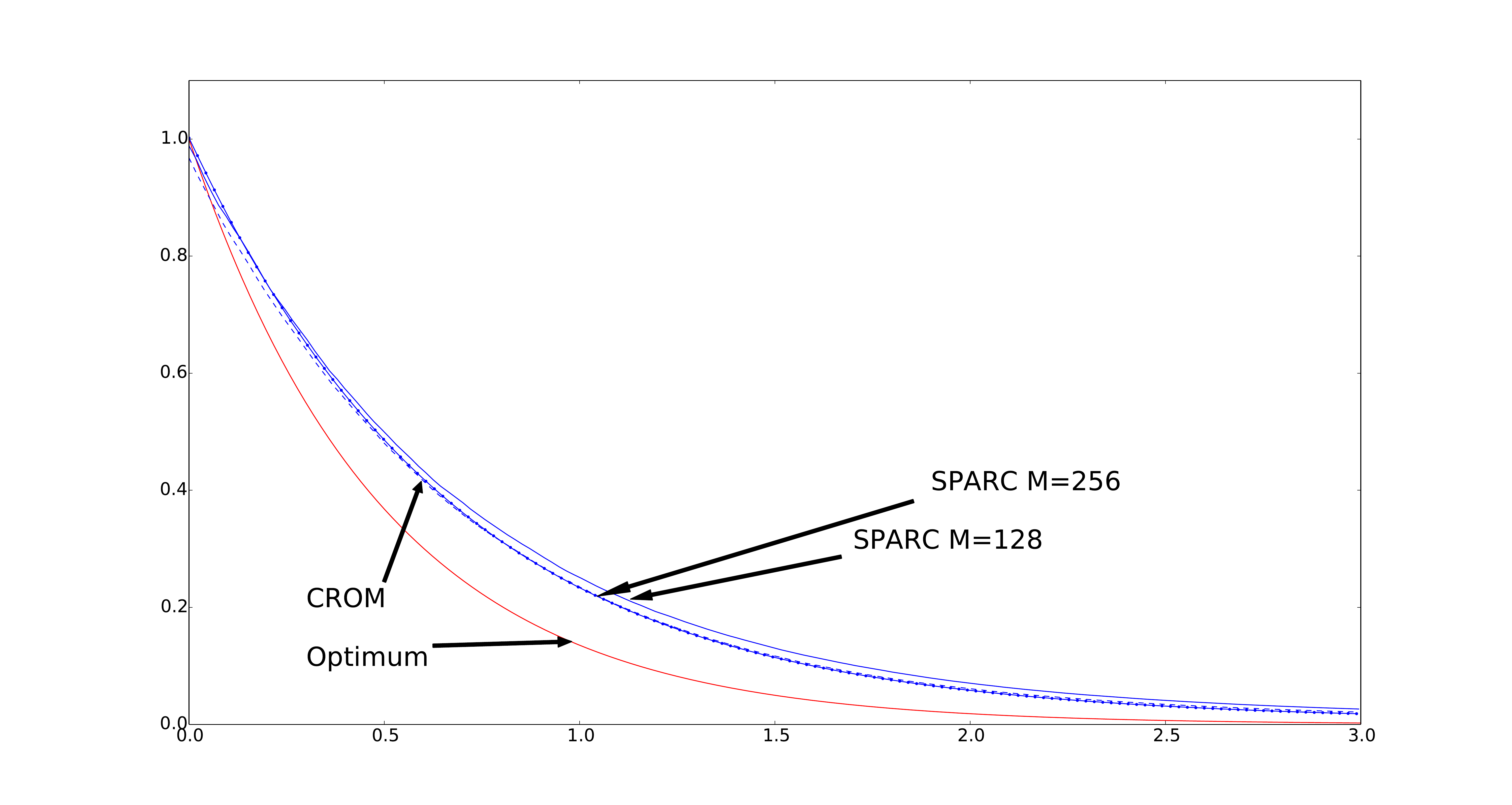}
\caption{Distortion-rate curves of CROM and SPARC where $n=256$.}
\label{fig:Comparison to SPARC}
\end{subfigure}~
\begin{subfigure}[t]{0.45\textwidth}
\centering
\includegraphics[width = \textwidth]{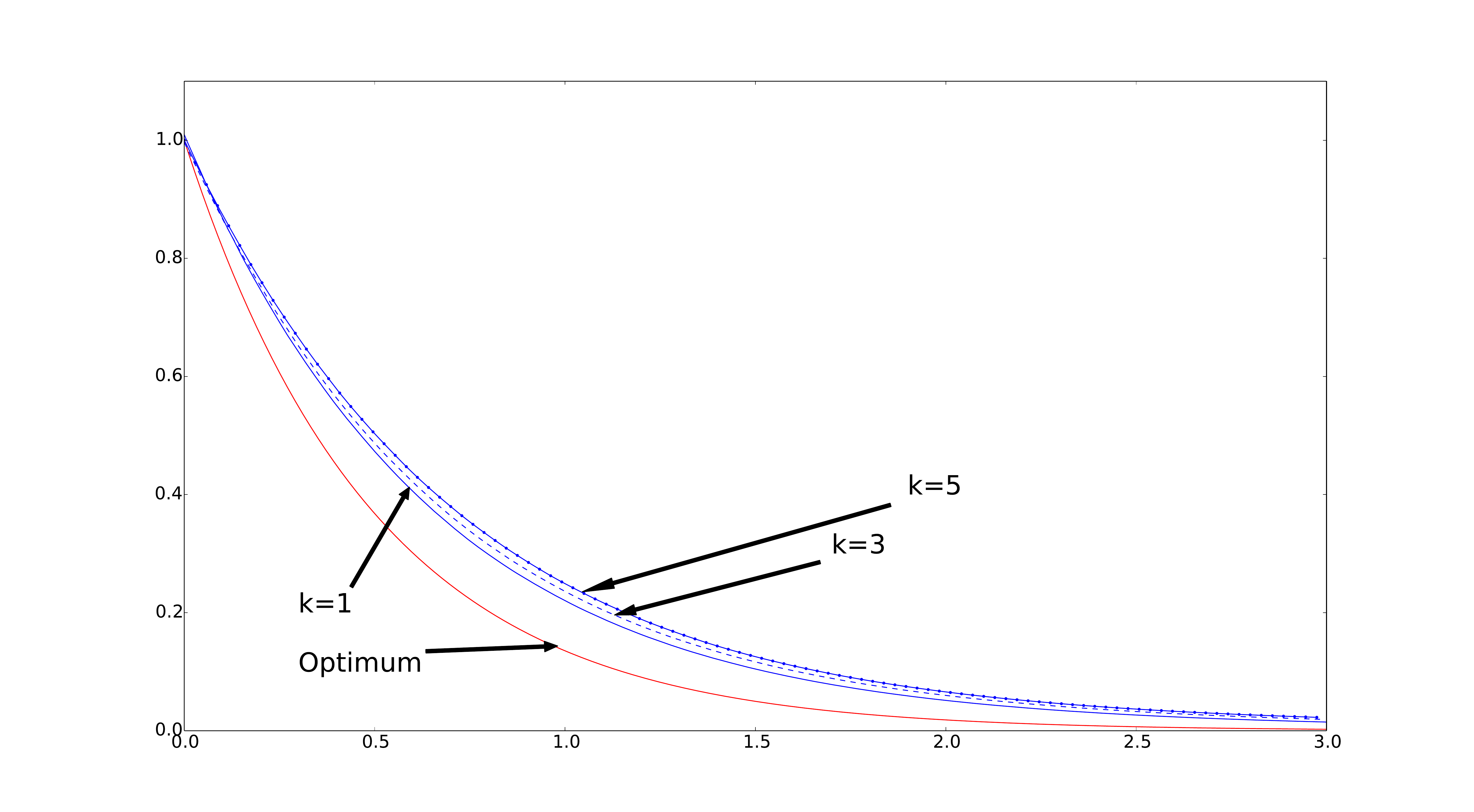}
\caption{Distortion-rate curves for $k=1,3,5$ where $n=1024$.}
\label{fig:different k n1024}
\end{subfigure}\\

\begin{subfigure}[t]{0.45\textwidth}
\centering
\includegraphics[width = \textwidth]{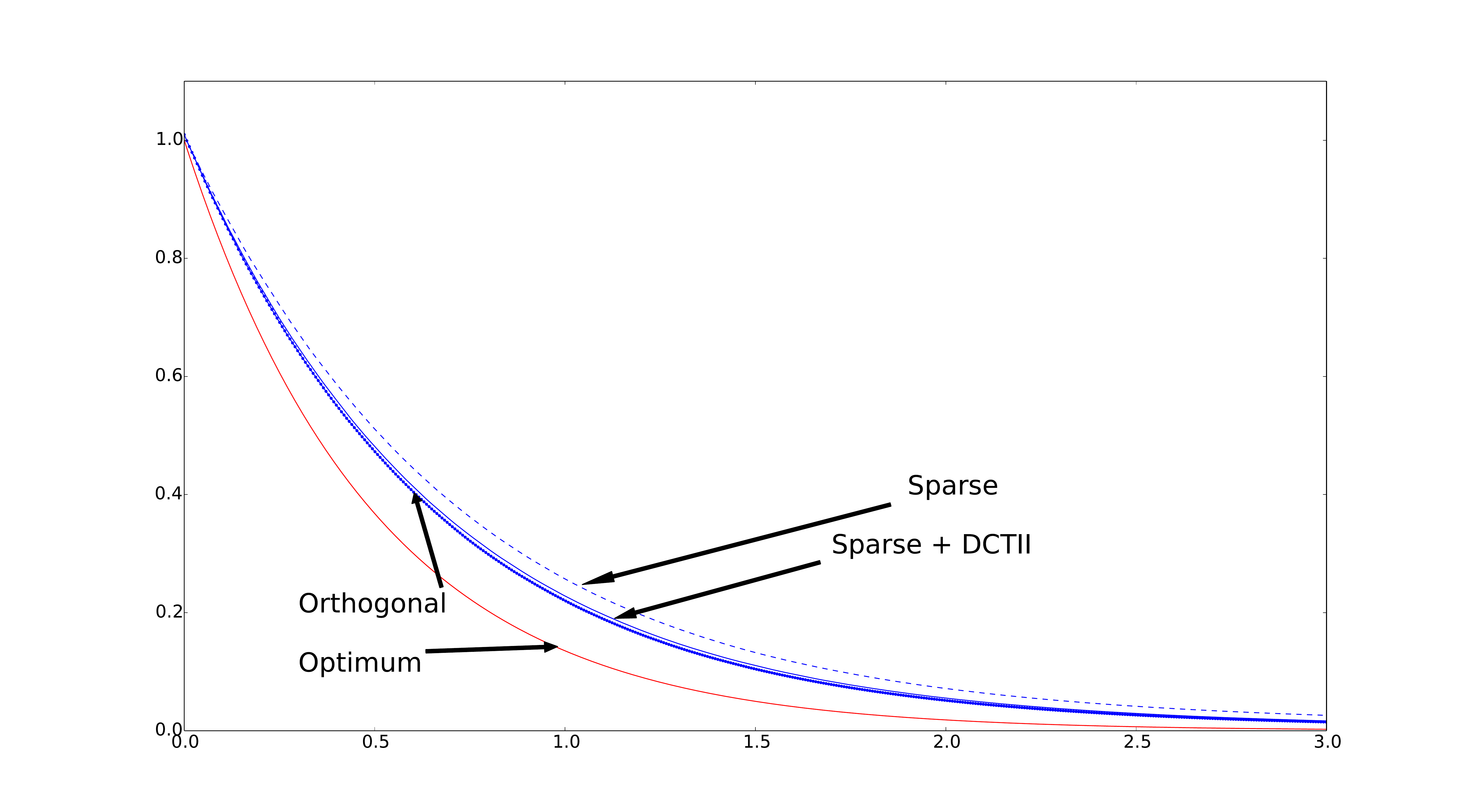}
\caption{Distortion-rate curves for different matrix constructions where $n=1024$.}
\label{fig:structured_matrices}
\end{subfigure}~
\begin{subfigure}[t]{0.45\textwidth}
\centering
\includegraphics[width = \textwidth]{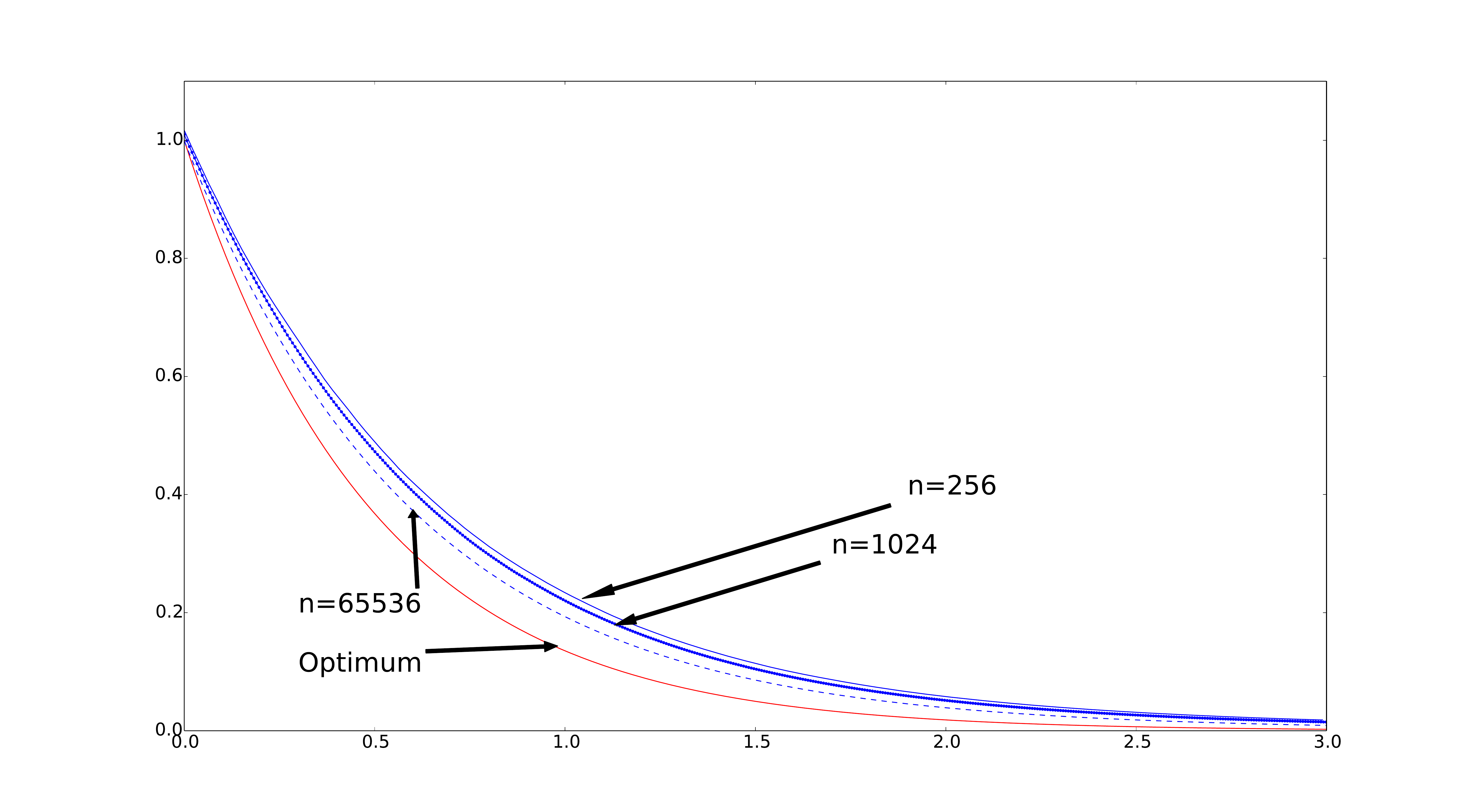}
\caption{Distortion-rate curves for $n=256,1024,65536$ where $k=1$.}
\label{fig:different n}
\end{subfigure}
\caption{Distortion-rate curves of CROM and SPARC. $x$-axis shows the rate in nats, and the $y$-axis represents the average distortion.}
\end{figure}

\section{Channel Coding Dual}\label{sec:Channel Coding Dual}
In \cite{polyanskiy2011minimum}, we can find a dual result in the Gaussian channel coding problem. In this section, we briefly review the idea of \cite{polyanskiy2011minimum} (with slightly changed notation). Consider the AWGN channel $Y_i = X_i+Z_i$ where $Z^n$ is an i.i.d.\ standard normal random vector. Suppose the number of messages is $n$, i.e., the rate of the scheme is $R_n = \frac{\log n}{n}$ nats per channel use. Based on message $m\in\{1,2,\ldots, n\}$, the encoder simply sends $X^n$ where $X_m = (1+\epsilon_n)\sqrt{2\log n}$ and $X_i = -(1+\epsilon_n)\frac{\sqrt{2\log n}}{n-1}$ if $i\neq m$. Then, the decoder finds the index of the maximum value of $Y^n$ and recovers the message, i.e., $\hat{m} = \arg\max_{1\leq i\leq n} Y_i$. The average power that the encoder uses is $P_n = 2(1+\epsilon_n)^2\frac{1}{n-1} \log n$. We will specify $\epsilon_n$ such that $\lim_{n\rightarrow\infty}\epsilon_n = 0$.

Before considering the probability of error $P_e^{(n)}$, let us introduce the following useful lemma. 
\begin{lemma}
Let $Z^n$ be an i.i.d.\ standard normal random vector, then
\begin{align}
\displaystyle\Pr{\max_{1\leq i\leq n} Z_i>\sqrt{2\log n}} \leq \frac{1}{\sqrt{\log n}}.
\end{align}
\end{lemma}

\begin{proof}
\begin{align}
\Pr{\max_{1\leq i\leq n} Z_i>\sqrt{2\log n}} =& 1- \Phi(\sqrt{2\log n})^n\\
=&1- (1-Q(a))^n\\
\leq & nQ(a)\\
\leq & n \frac{1}{\sqrt{2\log n}} \frac{1}{\sqrt{2\pi}} \exp\left(-\frac{2\log n}{2}\right)\\
\leq & \frac{1}{\sqrt{\log n}},
\end{align}
where $\Phi(x)$ is a standard normal cumulative distribution function and $Q(x) = 1-\Phi(x)$. We used the fact that $Q(x)\leq \frac{1}{x}f(x)$ where $f(x)$ is a probability density function of standard normal random variable.
\end{proof}

Now we are ready to bound $P_e^{(n)}$. Without loss of generality, we can assume that $m=1$.
\begin{align}
P_e^{(n)} =& \Pr{Y_1<\max_{2\leq i\leq n} Y_i}\\
=&\Pr{(1+\epsilon_n)\sqrt{2\log n} +Z_1 < -(1+\epsilon_n)\frac{\sqrt{2\log n}}{n-1}+\max_{2\leq i\leq n} Z_i}\\
=&\Pr{\frac{n}{n-1}(1+\epsilon_n)\sqrt{2\log n} +Z_1 <\max_{2\leq i\leq n} Z_i}\\
\leq&\Pr{\sqrt{2\log n} <\max_{2\leq i\leq n} Z_i}+\Pr{\frac{1+n\epsilon_n}{n-1}\sqrt{2\log n} +Z_1 <0}\\
\leq&\frac{1}{\sqrt{\log n}}+\Pr{\frac{1+n\epsilon_n}{n-1}\sqrt{2\log n} +Z_1 <0}.
\end{align}
If we choose $\epsilon_n$ such that $\frac{1+n\epsilon_n}{n-1} = \left(\log n\right)^{-1/3}$, then $\frac{1+n\epsilon_n}{n-1}\sqrt{2\log n}$ goes to infinity as $n$ grows. Therefore, 
\begin{align}
\lim_{n\rightarrow \infty} P_e^{(n)} = 0.
\end{align}

Since $P_n$ converges to zero as $n$ grows, we can approximate the capacity by $C(P_n) = \frac{1}{2}\log (1+P_n)\approx \frac{P_n}{2} = (1+\epsilon_n)^2 \frac{\log n}{n-1}$. It is clear that $\frac{R_n}{C(P_n)}$ converges to one as $n$ grows, i.e.,
\begin{align}
\lim_{n\rightarrow \infty} \frac{R_n}{C(P_n)} = 1.
\end{align}
This is reminiscent of the definition of a zero-rate optimal scheme in the source coding problem. We can say that this scheme is zero-rate optimal in the channel coding setting. We further note that the encoding and decoding can be done in almost linear time, and essentially no  extra information needs to be stored.

However, unlike CROM, we could not find an iterative scheme building on this zero-rate one that achieves reliable communication at a positive rate. The main challenge is that the tail behavior on the left side is very different from the right side. In the source coding problem, a small maximum value (which corresponds to the left tail) yields an error, while it is a large maximum value (which corresponds to the right tail) that yields an error in the channel coding problem. More precisely, the cumulative distribution function of the maximum of Gaussian random variables converges to $\exp\left(-e^{-x}\right)$ with normalizing constants. This function decays double-exponentially as $x$ decreases, which allows a small cumulative error for our iterative scheme CROM. However, $\exp\left(-e^{-x}\right)$ converges to one only exponentially as $x$ grows. Therefore, in the similar channel coding scheme, the cumulative error does not remain negligible when we employ the scheme iteratively. We believe that for similar reasons a channel coding analog of SPARC with efficient encoding would not work.

Note that Erez et al.\ discussed rateless coding for Gaussian channels \cite{erez2012rateless}. The goal of the paper ``Rateless Coding for Gaussian Channels seems design a channel code where the transmitter can be blind to the channel gain and the variance of the noise. Note that the proposed rateless code requires the base code that achieves the capacity. On the other hand, we would like to design a concrete coding scheme that achieves the channel capacity when the channel information is known.

\section{Proofs}\label{sec:Proofs}

\subsection{Extreme Value of Gaussian Random Variables}
Before providing proofs, consider the following lemma which shows the probabilistic bound of $Z_{(i)}$ when $Z^n$ is an i.i.d.\ standard normal random vector.
\begin{lemma}\label{lem:i-th order statistics}
Let $\epsilon>0$. If positive integers $n$ and $i$ satisfy $0\leq \frac{1}{n-i+1}\log\frac{n^{i-1}}{\epsilon} \leq 1$, then
\begin{align}
\Pr{Z_{(i)}<\Phi^{-1}\left(1-\frac{1}{n-i+1}\log\frac{n^{i-1}}{\epsilon}\right)} \leq \epsilon,
\end{align}
where $\Phi(x) = \int_{-\infty}^x \frac{1}{\sqrt{2\pi}}e^{-\frac{z^2}{2}}dz$ is a standard normal cumulative distribution function.
\end{lemma}

\begin{proof}
Since $\Phi(Z_1),\Phi(Z_2),\ldots,\Phi(Z_n)$ are i.i.d.\ uniform random variables, $\Phi(Z_{(i)})$ can be considered as the $i$-th largest value of an $n$ dimensional i.i.d.\ uniform random vector. The probability density function of $\Phi(Z_{(i)})$ is $\frac{n!}{(n-i)!(i-1)!} x^{n-i}(1-x)^{i-1}$. Therefore,
\begin{align}
\Pr{Z_{(i)}<\Phi^{-1}\left(1-\frac{1}{n-i+1}\log\frac{n^{i-1}}{\epsilon}\right)} = & \Pr{\Phi\left(Z_{(i)}\right)<1-\frac{1}{n-i+1}\log\frac{n^{i-1}}{\epsilon}}\\
=& \int_0^{1-\frac{1}{n-i+1}\log\frac{n^{i-1}}{\epsilon}}\frac{n!}{(n-i)!(i-1)!} x^{n-i}(1-x)^{i-1} dx\\
\leq& \int_0^{1-\frac{1}{n-i+1}\log\frac{n^{i-1}}{\epsilon}}\frac{n!}{(n-i)!(i-1)!} x^{n-i} dx\\
=& \frac{n!}{(n-i+1)!(i-1)!} \left(1-\frac{1}{n-i+1}\log\frac{n^{i-1}}{\epsilon}\right)^{n-i+1} \\
\leq&n^{i-1}\exp\left( -\log\frac{n^{i-1}}{\epsilon}\right) \\
=&\epsilon.
\end{align}
This concludes the proof.
\end{proof}

\subsection{Proof of Theorem \ref{thm:optimum zero rate}}\label{sec:Proof of optimum zero rate in excess distortion probability}
In the proof, we use $\alpha=\alpha_n$ for simplicity. By the definition of $\hat{X}^n$, we have
\begin{align}
\norm{X^n-\hat{X}^n}^2=\norm{X^n}^2+\frac{2k_n\alpha}{n-k_n}\sum_{i=1}^n X_i - \frac{2n\alpha}{n-k_n}\sum_{i=1}^{k_n} X_{(i)}+\frac{nk_n}{n-k_n}\alpha^2.
\end{align}
Let $\gamma_n$ and $\delta_n$ be positive real numbers where we specify their values later. Then,
\begin{align}
&\Pr{\norm{X^n-\hat{X}^n}^2>n(1+\gamma_n-\delta_n)}\nonumber\\
&=\Pr{\norm{X^n}^2+\frac{2k_n\alpha}{n-k_n}\sum_{i=1}^n X_i - \frac{2n\alpha}{n-k_n}\sum_{i=1}^{k_n} X_{(i)}+\frac{nk_n}{n-k_n}\alpha^2>n(1+\gamma_n-\delta_n)}\\
&\leq\Pr{\norm{X^n}^2 >n(1+\gamma_n)}+ \Pr{\frac{2k_n\alpha}{n-k_n}\sum_{i=1}^n X_i - \frac{2n\alpha}{n-k_n}\sum_{i=1}^{k_n} X_{(i)}+\frac{nk_n}{n-k_n}\alpha^2>-n\delta_n}.\label{eq:two terms of thm1 proof first step}
\end{align}
Consider the first term of \eqref{eq:two terms of thm1 proof first step}. Let $\gamma_n=\sqrt{\frac{2}{n}}Q^{-1}\left(\epsilon-\frac{15}{\sqrt{n}}-\frac{2}{n}\right)$, then we have
\begin{align}
\Pr{\norm{X^n}^2 >n(1+\gamma_n)}&\leq  Q\left(\sqrt{\frac{n}{2}}\gamma_n\right)+\frac{15}{\sqrt{n}}\label{eq:BerryEseen}\\
&=\epsilon-\frac{2}{n}.\label{eq:bound the first term of thm1 proof}
\end{align}
In \eqref{eq:BerryEseen}, we used Berry-Esseen theorem \cite{berry1941accuracy}:
\begin{align}
\sup_x \left|\Pr{\frac{\sum_{i=1}^n (X_i^2 - 1)}{\sigma\sqrt{n}} > x} - Q(x) \right| < \frac{\rho}{\sqrt{n}},
\end{align}
where $\sigma^2 = \E{X^4}-\E{X^2}^2 = 2$ and $\rho = \E{X_1^6} = 15$.

Consider the second term of \eqref{eq:two terms of thm1 proof first step}. 
\begin{align}
&\Pr{\frac{2k_n\alpha}{n-k_n}\sum_{i=1}^n X_i - \frac{2n\alpha}{n-k_n}\sum_{i=1}^{k_n} X_{(i)}+\frac{nk_n}{n-k_n}\alpha^2>-n\delta_n}\nonumber\\
&=\Pr{\frac{1}{n}\sum_{i=1}^n X_i - \frac{1}{k_n}\sum_{i=1}^{k_n} X_{(i)}+\frac{\alpha}{2}>-\frac{n-k_n}{2k_n\alpha}\delta_n}\\
&=\Pr{\frac{1}{k_n}\sum_{i=1}^{k_n} X_{(i)}-\frac{1}{n}\sum_{i=1}^n X_i <\frac{\alpha}{2}+\frac{n-k_n}{2k_n\alpha}\delta_n}.
\end{align}
Let $\alpha = \sqrt{\frac{n-k_n}{k_n}\delta_n}=p_n-q_n$, where
\begin{align}
p_n =& \Phi^{-1}\left(1-\frac{1}{n-k_n+1}\log n^{k_n}\right)\\
=& Q^{-1}\left(\frac{k_n}{n-k_n+1}\log n\right)\label{eq:p_n}\\
q_n =&\frac{1}{\sqrt{n}}Q^{-1}\left(\frac{1}{n}\right).\label{eq:q_n}
\end{align}
Then, we have
\begin{align}
&\Pr{\frac{2k_n\alpha}{n-k_n}\sum_{i=1}^n X_i - \frac{2n\alpha}{n-k_n}\sum_{i=1}^{k_n} X_{(i)}+\frac{nk_n}{n-k_n}\alpha^2>-n\delta_n}\nonumber\\
&=\Pr{\frac{1}{k_n}\sum_{i=1}^{k_n} X_{(i)}-\frac{1}{n}\sum_{i=1}^n X_i <p_n-q_n}\\
&\leq\Pr{\frac{1}{k_n}\sum_{i=1}^{k_n} X_{(i)} <p_n}+\Pr{\frac{1}{n}\sum_{i=1}^n X_i >q_n}\\
&\leq\Pr{X_{(k_n)} <p_n}+\Pr{\frac{1}{n}\sum_{i=1}^n X_i >q_n}.
\end{align}
By Lemma \ref{lem:i-th order statistics}, $\Pr{X_{(k_n)} <p_n}\leq\frac{1}{n}$. Since $\frac{1}{n}\sum_{i=1}^n X_i$ has a Gaussian distribution with zero mean and variance $\frac{1}{n}$,
\begin{align}
\Pr{\frac{1}{n}\sum_{i=1}^n X_i >q_n}=&\frac{1}{n}.
\end{align}
Therefore,
\begin{align}
\Pr{\frac{2k_n\alpha}{n-k_n}\sum_{i=1}^n X_i - \frac{2n\alpha}{n-k_n}\sum_{i=1}^{k_n} X_{(i)}+\frac{nk_n}{n-k_n}\alpha^2>-n\delta_n}\leq\frac{2}{n}.
\end{align}
With \eqref{eq:bound the first term of thm1 proof}, we have
\begin{align}
\Pr{\norm{X^n-\hat{X}^n}^2>n(1+\gamma_n-\delta_n)}\leq \epsilon.
\end{align}

Now, let consider the bound on $1+\gamma_n-\delta_n$. It is clear that the inequality $\frac{\sqrt{2\log \frac{n-k_n+1}{k_n\log^3 n}}}{1+2\log \frac{n-k_n+1}{k_n\log^3 n}} \frac{1}{\sqrt{2\pi}}  \log^2 n>1$ holds for large enough $n$, and therefore
\begin{align}
Q\left(\sqrt{2\log \frac{n-k_n+1}{k_n\log^3 n}}\right)\geq& \frac{\sqrt{2\log \frac{n-k_n+1}{k_n\log^3 n}}}{1+2\log \frac{n-k_n+1}{k_n\log^3 n}} \frac{1}{\sqrt{2\pi}} \frac{k_n\log^3 n}{n-k_n+1}\\
\geq &\frac{k_n\log n}{n-k_n+1},
\end{align}
which implies
\begin{align}
p_n = Q^{-1}\left(\frac{k_n}{n-k_n+1}\log n\right)\geq \sqrt{2\log \frac{n-k_n+1}{k_n\log^3 n}}.
\end{align}
On the other hand, it is not hard to show that
\begin{align}
q_n =\frac{1}{\sqrt{n}}Q^{-1}\left(\frac{1}{n}\right)\leq  \sqrt{\frac{2\log\frac{n}{2}}{n}}.
\end{align}

Now, we are ready to bound $D_n = 1+\gamma_n-\delta_n$. Since $R_n = \frac{1}{n}\log{n \choose k_n}$, we have
\begin{align}
D_n =& 1+\gamma_n - \frac{k_n}{n-k_n}(p_n-q_n)^2\\
=&1+\sqrt{\frac{2}{n}}Q^{-1}\left(\epsilon-\frac{15}{\sqrt{n}}-\frac{2}{n}\right)- \frac{k_n}{n-k_n}\left(Q^{-1}\left(\frac{k_n}{n-k_n+1}\log n\right)-\frac{1}{\sqrt{n}}Q^{-1}\left(\frac{1}{n}\right)\right)^2\\
\leq&1+ \sqrt{\frac{2}{n}}Q^{-1}(\epsilon) +O\left(\frac{1}{n}\right)- \frac{k_n}{n-k_n}\left(\sqrt{2\log \frac{n-k_n+1}{k_n\log^3 n}}-\sqrt{\frac{2\log\frac{n}{2}}{n}}\right)^2\\
=&1+ \sqrt{\frac{2}{n}}Q^{-1}(\epsilon) - 2R_n+O\left(\frac{k_n\log \log n}{n}\right).
\end{align}
This concludes the proof.

\subsection{Proof of Theorem \ref{thm:infinite refinability}}\label{sec: proof of infinite refinability}
Throughout the proof, we will let $\sigma^2=1$ and use $L$ instead of $L_n$ for simplicity. Also, instead of choosing specific orthogonal matrices $A_1,\ldots, A_{L+1}$, we employ a randomization argument. More precisely, we assume that $\Ab_{1},\Ab_{2},\ldots,\Ab_{i+1}$ are drawn i.i.d.\ $\mathUO$ and show that equation \eqref{eq:infinite refinability thm} holds when the probability is averaged over this ensemble of random matrices. Let $S_i = \norm{\Xb{i}}$ and $\Xb{i} = S_i\Bb{i}$ where $\Bb{i}$ is uniformly distributed on the $n$-dimensional unit sphere and independent to $S_i$. Since we draw random matrices independently, random variables $\Bb{1},\ldots,\Bb{L+1}$ are also independent. Recall \eqref{eq:after ith iteration} and \eqref{eq:xhat after ith iteration}, we have $\norm{X^n - \Xbh{i}}^2 =  \norm{\Xb{i+1}}^2 = S_{i+1}^2$, and this implies that the distortion after the $i$-th iteration coincides with $S_{i+1}^2$ divided by $n$. We further let $\tilde{S}$ be a chi-distributed random variable with degrees of freedom $n$ and independent to all $\Bb{i}$'s, i.e., $\tilde{S}^2 \sim \chi^2(n)$. Using union bound, we can obtain an upper bound on the excess distortion probability.
\begin{align}
&\Pr{\frac{1}{\sqrt{n}}S_{i+1}>e^{-\frac{i}{L}R} + e^{\frac{i}{L}R}\gamma_n\mbox{ for some $0\leq i\leq L$}} \nonumber\\
&\leq  \Pr{\frac{1}{\sqrt{n}}S_1>1+\gamma_n \mbox{ or }\frac{1}{\sqrt{n}}\tilde{S}>\sqrt{1+\gamma_n}}\nonumber\\
&+\sum_{i=1}^L \Pr{\frac{1}{\sqrt{n}}S_{i+1}>e^{-\frac{i}{L}R} + e^{\frac{i}{L}R}\gamma_n,   \frac{1}{\sqrt{n}}S_{j+1}\leq e^{-\frac{j}{L}R} + e^{\frac{j}{L}R}\gamma_n\mbox{ for all $j<i$}, \mbox{ and }\frac{1}{\sqrt{n}}\tilde{S}\leq \sqrt{1+\gamma_n}}\\
&\leq  \Pr{\frac{1}{\sqrt{n}}S_1>1+\gamma_n \mbox{ or }\frac{1}{\sqrt{n}}\tilde{S}>\sqrt{1+\gamma_n}}\nonumber\\
&+\sum_{i=1}^L \Pr{\frac{1}{\sqrt{n}}S_{i+1}>e^{-\frac{i}{L}R} + e^{\frac{i}{L}R}\gamma_n, \frac{1}{\sqrt{n}}S_{i}\leq e^{-\frac{i-1}{L}R} + e^{\frac{i-1}{L}R}\gamma_n, \mbox{ and }\frac{1}{\sqrt{n}}\tilde{S}\leq \sqrt{1+\gamma_n}}
\end{align}
From the definition of $\Xb{i+1}$, we have
\begin{align}
S_{i+1}^2=&\norm{\Xb{i+1}}^2\\
=& \norm{\Xb{i}-\alpha_i \Ub{i}}^2\\
=& \norm{\Xb{i}}^2 + \alpha_i^2-2\alpha_i\left(-\sqrt{\frac{k_n}{n(n-k_n)}} \mathbf{1} + \sqrt{\frac{n}{(n-k_n)k_n}}\mb{i}\right)^T \Xb{i},\label{eq:complicated recursion}
\end{align}
where $\left(\mb{i}\right)^T \Xb{i}$ is a sum of $k_n$ largest value of $\Xb{i}$. Let $V_i = \left(-\sqrt{\frac{k_n}{n(n-k_n)}} \mathbf{1} + \sqrt{\frac{n}{(n-k_n)k_n}}\mb{i}\right)^T \Bb{i}$, then $V_i$ and $S_i$ are independent. We can now rewrite \eqref{eq:complicated recursion} as
\begin{align}
S_{i+1}^2 = S_i^2 + \alpha_i^2 - 2\alpha_i S_i V_i.
\end{align}
It is not hard to show that $S_i^2+\alpha_i^2 - 2\alpha_iS_iV_i$ is an increasing function in $S_i$ when $\frac{1}{\sqrt{n}}S_{i+1}>e^{-\frac{i}{L}R} + e^{\frac{i}{L}R}\gamma_n$ and $\frac{1}{\sqrt{n}}S_i\leq e^{-\frac{i-1}{L}R} + e^{\frac{i-1}{L}R}\gamma_n$. Therefore,
\begin{align}
S_{i+1}^2=S_i^2 + \alpha_i^2 - 2\alpha_i S_i V_i\leq  n\left(e^{-\frac{i-1}{L}R} + e^{\frac{i-1}{L}R}\gamma_n\right)^2+\alpha_i^2 -2\sqrt{n}\alpha_i\left(e^{-\frac{i-1}{L}R} + e^{\frac{i-1}{L}R}\gamma_n\right)V_i,
\end{align}
which is equivalent to
\begin{align}
 \frac{n\left(e^{-\frac{i-1}{L}R} + e^{\frac{i-1}{L}R}\gamma_n\right)^2+\alpha_i^2 - n\left(e^{-\frac{i}{L}R} + e^{\frac{i}{L}R}\gamma_n\right)^2}{2\sqrt{n}\alpha_i\left(e^{-\frac{i-1}{L}R} + e^{\frac{i-1}{L}R}\gamma_n\right)} >V_i.
\end{align}
This implies
\begin{align}
&\Pr{\frac{1}{\sqrt{n}}S_{i+1}>e^{-\frac{i}{L}R} + e^{\frac{i}{L}R}\gamma_n, \frac{1}{\sqrt{n}}S_{i}\leq e^{-\frac{i-1}{L}R} + e^{\frac{i-1}{L}R}\gamma_n, \mbox{ and }\frac{1}{\sqrt{n}}\tilde{S}\leq \sqrt{1+\gamma_n}}\nonumber\\
&= \Pr{ \frac{n\left(e^{-\frac{i-1}{L}R} + e^{\frac{i-1}{L}R}\gamma_n\right)^2+\alpha_i^2 - n\left(e^{-\frac{i}{L}R} + e^{\frac{i}{L}R}\gamma_n\right)^2}{2\sqrt{n}\alpha_i\left(e^{-\frac{i-1}{L}R} + e^{\frac{i-1}{L}R}\gamma_n\right)} >V_i \mbox{ and }\frac{1}{\sqrt{n}}\tilde{S}\leq \sqrt{1+\gamma_n}}.
\end{align}
Recall that we took 
\begin{align}
\alpha_i = \sqrt{n\left(1-e^{-\frac{2}{L}R}\right)\left(e^{-\frac{i-1}{L}R}+e^{\frac{i-1}{L}R}\gamma_n\right)\left(e^{-\frac{i-1}{L}R}-e^{\frac{i-1}{L}R}\gamma_n\right)},
\end{align}
and it can be easily shown that
\begin{align}
\frac{n\left(e^{-\frac{i-1}{L}R} + e^{\frac{i-1}{L}R}\gamma_n\right)^2+\alpha_i^2 - n\left(e^{-\frac{i}{L}R} + e^{\frac{i}{L}R}\gamma_n\right)^2}{2\sqrt{n}\alpha_i\left(e^{-\frac{i-1}{L}R} + e^{\frac{i-1}{L}R}\gamma_n\right)}&=\frac{2\alpha_i^2}{2\sqrt{n}\alpha_i\left(e^{-\frac{i-1}{L}R} + e^{\frac{i-1}{L}R}\gamma_n\right)}\\
&\leq\sqrt{\left(1-e^{-\frac{2}{L}R}\right)\frac{e^{-\frac{i-1}{L}R}-e^{\frac{i-1}{L}R}\gamma_n}{e^{-\frac{i-1}{L}R}+e^{\frac{i-1}{L}R}\gamma_n}}\\
&\leq\sqrt{\left(1-e^{-\frac{2}{L}R}\right)\frac{1-\gamma_n}{1+\gamma_n}}.
\end{align}
Thus, we have
\begin{align}
&\Pr{\frac{1}{\sqrt{n}}S_{i+1}>e^{-\frac{i}{L}R} + e^{\frac{i}{L}R}\gamma_n,\mbox{ and }  \frac{1}{\sqrt{n}}S_{i}\leq e^{-\frac{i-1}{L}R} + e^{\frac{i-1}{L}R}\gamma_n, \mbox{ and }\frac{1}{\sqrt{n}}\tilde{S}\leq \sqrt{1+\gamma_n}}\nonumber\\
&\leq \Pr{\sqrt{\left(1-e^{-\frac{2}{L}R}\right)\frac{1-\gamma_n}{1+\gamma_n}} > V_i \mbox{ and } \frac{1}{\sqrt{n}}\tilde{S}\leq \sqrt{1+\gamma_n}}\\
&\leq \Pr{\sqrt{n\left(1-e^{-\frac{2}{L}R}\right)(1-\gamma_n)} > \tilde{S}V_i}\\
&\leq \Pr{\sqrt{\frac{2nR}{L}(1-\gamma_n)} > \tilde{S}V_i }.
\end{align}

Since $\Bb{i}$ is uniformly distributed on a unit sphere and it is independent of $\tilde{S}$, we have $\tilde{S}\Bb{i} \stackrel{(d)}{=} {\bf Z}$ where ${\bf Z}$ is an $n$ dimensional i.i.d.\ standard normal random vector. Furthermore,
\begin{align}
\tilde{S}V_i \stackrel{(d)}{\equiv}&\left(-\sqrt{\frac{k_n}{n(n-k_n)}} \mathbf{1} + \sqrt{\frac{n}{(n-k_n)k_n}}\mb{1}\right)^T {\bf Z}\\
=& \sqrt{\frac{nk_n}{n-k_n}}\left(\frac{1}{k_n}\sum_{i=1}^{k_n}Z_{(i)} - \frac{1}{n} \sum_{i=1}^n Z_i \right).
\end{align}

If we have $\gamma_n \geq 1- (p_n-q_n)^2 \frac{L}{2nR}\frac{nk_n}{n-k_n}$, where $p_n$ and $q_n$ were defined as \eqref{eq:p_n} and \eqref{eq:q_n}, then we can apply the similar technique from the proof of Theorem \ref{thm:optimum zero rate}. I.e., 
\begin{align}
&\Pr{\frac{1}{\sqrt{n}}S_{i+1}>e^{-\frac{i}{L}R} + e^{\frac{i}{L}R}\gamma_n,\mbox{ and }  \frac{1}{\sqrt{n}}S_{i}\leq e^{-\frac{i-1}{L}R} + e^{\frac{i-1}{L}R}\gamma_n, \mbox{ and }\frac{1}{\sqrt{n}}\tilde{S}\leq \sqrt{1+\gamma_n}}\nonumber\\
&\leq \Pr{\sqrt{\frac{2nR}{L}(1-\gamma_n)} > \tilde{S}V_i }\\
&\leq\Pr{p_n-q_n >\frac{1}{k_n}\sum_{i=1}^{k_n}Z_{(i)} - \frac{1}{n} \sum_{i=1}^n Z_i }\\
&\leq \frac{2}{n}.
\end{align}
Recall that $p_n\geq \sqrt{2\log\frac{n-k_n+1}{k_n\log^3 n}}$ and $q_n \leq \sqrt{\frac{2\log\frac{n}{2}}{n}}$. Therefore, it is easy to check that
\begin{align}
1- (p_n-q_n)^2 \frac{L}{2nR}\frac{nk_n}{n-k_n}&\leq 1- \left(\sqrt{2\log\frac{n-k_n+1}{k_n\log^3 n}}-\sqrt{\frac{2\log\frac{n}{2}}{n}}\right)^2 \frac{k_n}{2\log{n\choose k}} \\
&= O\left(\frac{\log\log n}{\log n}\right).
\end{align} 

Firstly, if $\gamma_n$ is equal to any constant $\gamma>0$, due to the stationarity of the source, we have 
\begin{align}
\lim_{n\rightarrow\infty}\Pr{\frac{1}{\sqrt{n}}S_1>1+\gamma_n \mbox{ or }\frac{1}{\sqrt{n}}\tilde{S}>\sqrt{1+\gamma_n}} = 0.
\end{align}
Therefore,
\begin{align}
&\lim_{n\rightarrow\infty}\Pr{\frac{1}{\sqrt{n}}S_{i+1}>e^{-\frac{i}{L}R} + e^{\frac{i}{L}R}\gamma_n\mbox{ for some $0\leq i\leq L$}}\nonumber\\
&=\lim_{n\rightarrow\infty}\Pr{\frac{1}{\sqrt{n}}S_1>1+\gamma_n \mbox{ or }\frac{1}{\sqrt{n}}\tilde{S}>\sqrt{1+\gamma_n}} +\lim_{n\rightarrow\infty}\frac{2}{n}L_n\\
&=0.
\end{align}

Suppose the source is i.i.d.\ distributed with $\E{|X_1|^3}<\infty$, then we can let $\gamma_n = O\left(\frac{\log\log n}{\log n}\right)$ such that
\begin{align}
\gamma_n \geq 1- \left(\sqrt{2\log\frac{n-k_n+1}{k_n\log^3 n}}-\sqrt{\frac{2\log\frac{n}{2}}{n}}\right)^2 \frac{k_n}{2\log{n\choose k}}\label{eq:condition on gamma}
\end{align}
and still have
\begin{align}
\lim_{n\rightarrow\infty}\Pr{\frac{1}{\sqrt{n}}S_1>1+\gamma_n \mbox{ or }\frac{1}{\sqrt{n}}\tilde{S}>\sqrt{1+\gamma_n}} = 0.
\end{align}
We would like to point out that the right hand side of \eqref{eq:condition on gamma} is independent to the choice of $R$. Finally, it is clear that
\begin{align}
&\lim_{n\rightarrow\infty}\Pr{\frac{1}{\sqrt{n}}S_{i+1}>e^{-\frac{i}{L}R} + e^{\frac{i}{L}R}\gamma_n\mbox{ for some $0\leq i\leq L$}}\nonumber\\
&=\lim_{n\rightarrow\infty}\Pr{\frac{1}{\sqrt{n}}S_1>1+\gamma_n \mbox{ or }\frac{1}{\sqrt{n}}\tilde{S}>\sqrt{1+\gamma_n}} +\lim_{n\rightarrow\infty}\frac{2}{n}L_n\\
&=0.
\end{align}
This concludes the proof.

\section{Conclusions}\label{sec:Conclusions}
Our starting point (and inspiration for the subsequent main scheme and result) was an extremely simple scheme that achieves the optimum zero-rate distortion for the Gaussian source. We then generalized it to CROM, a lossy source coding scheme that simultaneously achieves the distortion-rate function of the Gaussian memoryless source for all rates when operating on any ergodic source.
The merit of CROM over classical random coding schemes is its low storage and computational complexity, as well as the fact that the encoding can be oblivious to the rate desired while the decoding is essentially sequential (sub-linear lookahead) and simultaneously achieves all points on the distortion-rate curve.

%


\bibliographystyle{IEEEtran}
\bibliography{../../AlbertRef}
\end{document}